\documentclass[letterpaper, 10 pt, conference]{ieeeconf}  % Comment this line out
\IEEEoverridecommandlockouts                              % This command is only
\usepackage{graphics} % for pdf, bitmapped graphics files
\usepackage{epsfig} % for postscript graphics files
\usepackage{times} % assumes new font selection scheme installed
\usepackage{amsmath} % assumes amsmath package installed
\usepackage{amssymb}  % assumes amsmath package installed\
\usepackage{subfig}
\usepackage{bbding}
\usepackage{pifont}
\usepackage{wasysym}
\usepackage{amssymb}
\usepackage{algorithm}
\usepackage{algorithmic}
\usepackage{xcolor}
\usepackage{booktabs}
\newtheorem{theorem}{Theorem}
\newtheorem{remark}[theorem]{Remark}

\DeclareMathOperator{\Tr}{Tr}
\DeclareMathOperator*{\argmin}{arg\,min}
\title{\LARGE \bf
Deep Graphic FBSDEs for Opinion Dynamics Stochastic Control
}

\author{Tianrong Chen, Ziyi Wang, Evangelos A. Theodorou% <-this % stops a space
\thanks{Georgia Institute of Technology, Atlanta, GA, USA.
        Email:{\tt\small tchen429@gatech.edu}}%
}

\begin{document}

\maketitle
\thispagestyle{empty}
\pagestyle{empty}
\newcommand{\norm}[1]{\left\lVert#1\right\rVert}
\newcommand{\T}{\mathsf{T}}
% Random variables
\def\reta{{\textnormal{$\eta$}}}
\def\ra{{\textnormal{a}}}
\def\rb{{\textnormal{b}}}
\def\rc{{\textnormal{c}}}
\def\rd{{\textnormal{d}}}
\def\re{{\textnormal{e}}}
\def\rf{{\textnormal{f}}}
\def\rg{{\textnormal{g}}}
\def\rh{{\textnormal{h}}}
\def\ri{{\textnormal{i}}}
\def\rj{{\textnormal{j}}}
\def\rk{{\textnormal{k}}}
\def\rl{{\textnormal{l}}}
% rm is already a command, just don't name any random variables m
\def\rn{{\textnormal{n}}}
\def\ro{{\textnormal{o}}}
\def\rp{{\textnormal{p}}}
\def\rq{{\textnormal{q}}}
\def\rr{{\textnormal{r}}}
\def\rs{{\textnormal{s}}}
\def\rt{{\textnormal{t}}}
\def\ru{{\textnormal{u}}}
\def\rv{{\textnormal{v}}}
\def\rw{{\textnormal{w}}}
\def\rx{{\textnormal{x}}}
\def\ry{{\textnormal{y}}}
\def\rz{{\textnormal{z}}}

% Random vectors
\def\rvepsilon{{\mathbf{\epsilon}}}
\def\rvtheta{{\mathbf{\theta}}}
\def\rva{{\mathbf{a}}}
\def\rvb{{\mathbf{b}}}
\def\rvc{{\mathbf{c}}}
\def\rvd{{\mathbf{d}}}
\def\rve{{\mathbf{e}}}
\def\rvf{{\mathbf{f}}}
\def\rvg{{\mathbf{g}}}
\def\rvh{{\mathbf{h}}}
\def\rvu{{\mathbf{i}}}
\def\rvj{{\mathbf{j}}}
\def\rvk{{\mathbf{k}}}
\def\rvl{{\mathbf{l}}}
\def\rvm{{\mathbf{m}}}
\def\rvn{{\mathbf{n}}}
\def\rvo{{\mathbf{o}}}
\def\rvp{{\mathbf{p}}}
\def\rvq{{\mathbf{q}}}
\def\rvr{{\mathbf{r}}}
\def\rvs{{\mathbf{s}}}
\def\rvt{{\mathbf{t}}}
\def\rvu{{\mathbf{u}}}
\def\rvv{{\mathbf{v}}}
\def\rvw{{\mathbf{w}}}
\def\rvx{{\mathbf{x}}}
\def\rvX{{\mathbf{X}}}
\def\rvy{{\mathbf{y}}}
\def\rvY{{\mathbf{Y}}}
\def\rvz{{\mathbf{z}}}

% Elements of random vectors
\def\erva{{\textnormal{a}}}
\def\ervb{{\textnormal{b}}}
\def\ervc{{\textnormal{c}}}
\def\ervd{{\textnormal{d}}}
\def\erve{{\textnormal{e}}}
\def\ervf{{\textnormal{f}}}
\def\ervg{{\textnormal{g}}}
\def\ervh{{\textnormal{h}}}
\def\ervi{{\textnormal{i}}}
\def\ervj{{\textnormal{j}}}
\def\ervk{{\textnormal{k}}}
\def\ervl{{\textnormal{l}}}
\def\ervm{{\textnormal{m}}}
\def\ervn{{\textnormal{n}}}
\def\ervo{{\textnormal{o}}}
\def\ervp{{\textnormal{p}}}
\def\ervq{{\textnormal{q}}}
\def\ervr{{\textnormal{r}}}
\def\ervs{{\textnormal{s}}}
\def\ervt{{\textnormal{t}}}
\def\ervu{{\textnormal{u}}}
\def\ervv{{\textnormal{v}}}
\def\ervw{{\textnormal{w}}}
\def\ervx{{\textnormal{x}}}
\def\ervy{{\textnormal{y}}}
\def\ervz{{\textnormal{z}}}

% Random matrices
\def\rmA{{\mathbf{A}}}
\def\rmB{{\mathbf{B}}}
\def\rmC{{\mathbf{C}}}
\def\rmD{{\mathbf{D}}}
\def\rmE{{\mathbf{E}}}
\def\rmF{{\mathbf{F}}}
\def\rmG{{\mathbf{G}}}
\def\rmH{{\mathbf{H}}}
\def\rmI{{\mathbf{I}}}
\def\rmJ{{\mathbf{J}}}
\def\rmK{{\mathbf{K}}}
\def\rmL{{\mathbf{L}}}
\def\rmM{{\mathbf{M}}}
\def\rmN{{\mathbf{N}}}
\def\rmO{{\mathbf{O}}}
\def\rmP{{\mathbf{P}}}
\def\rmQ{{\mathbf{Q}}}
\def\rmR{{\mathbf{R}}}
\def\rmS{{\mathbf{S}}}
\def\rmT{{\mathbf{T}}}
\def\rmU{{\mathbf{U}}}
\def\rmV{{\mathbf{V}}}
\def\rmW{{\mathbf{W}}}
\def\rmX{{\mathbf{X}}}
\def\rmY{{\mathbf{Y}}}
\def\rmZ{{\mathbf{Z}}}

% Elements of random matrices
\def\ermA{{\textnormal{A}}}
\def\ermB{{\textnormal{B}}}
\def\ermC{{\textnormal{C}}}
\def\ermD{{\textnormal{D}}}
\def\ermE{{\textnormal{E}}}
\def\ermF{{\textnormal{F}}}
\def\ermG{{\textnormal{G}}}
\def\ermH{{\textnormal{H}}}
\def\ermI{{\textnormal{I}}}
\def\ermJ{{\textnormal{J}}}
\def\ermK{{\textnormal{K}}}
\def\ermL{{\textnormal{L}}}
\def\ermM{{\textnormal{M}}}
\def\ermN{{\textnormal{N}}}
\def\ermO{{\textnormal{O}}}
\def\ermP{{\textnormal{P}}}
\def\ermQ{{\textnormal{Q}}}
\def\ermR{{\textnormal{R}}}
\def\ermS{{\textnormal{S}}}
\def\ermT{{\textnormal{T}}}
\def\ermU{{\textnormal{U}}}
\def\ermV{{\textnormal{V}}}
\def\ermW{{\textnormal{W}}}
\def\ermX{{\textnormal{X}}}
\def\ermY{{\textnormal{Y}}}
\def\ermZ{{\textnormal{Z}}}

% Vectors
\def\vzero{{\bm{0}}}
\def\vone{{\bm{1}}}
\def\vmu{{\bm{\mu}}}
\def\vtheta{{\bm{\theta}}}
\def\va{{\bm{a}}}
\def\vb{{\bm{b}}}
\def\vc{{\bm{c}}}
\def\vd{{\bm{d}}}
\def\ve{{\bm{e}}}
\def\vf{{\bm{f}}}
\def\vg{{\bm{g}}}
\def\vh{{\bm{h}}}
\def\vi{{\bm{i}}}
\def\vj{{\bm{j}}}
\def\vk{{\bm{k}}}
\def\vl{{\bm{l}}}
\def\vm{{\bm{m}}}
\def\vn{{\bm{n}}}
\def\vo{{\bm{o}}}
\def\vp{{\bm{p}}}
\def\vq{{\bm{q}}}
\def\vr{{\bm{r}}}
\def\vs{{\bm{s}}}
\def\vt{{\bm{t}}}
\def\vu{{\bm{u}}}
\def\vv{{\bm{v}}}
\def\vw{{\bm{w}}}
\def\vx{{\bm{x}}}
\def\vy{{\bm{y}}}
\def\vz{{\bm{z}}}

% Elements of vectors
\def\evalpha{{\alpha}}
\def\evbeta{{\beta}}
\def\evepsilon{{\epsilon}}
\def\evlambda{{\lambda}}
\def\evomega{{\omega}}
\def\evmu{{\mu}}
\def\evpsi{{\psi}}
\def\evsigma{{\sigma}}
\def\evtheta{{\theta}}
\def\eva{{a}}
\def\evb{{b}}
\def\evc{{c}}
\def\evd{{d}}
\def\eve{{e}}
\def\evf{{f}}
\def\evg{{g}}
\def\evh{{h}}
\def\evi{{i}}
\def\evj{{j}}
\def\evk{{k}}
\def\evl{{l}}
\def\evm{{m}}
\def\evn{{n}}
\def\evo{{o}}
\def\evp{{p}}
\def\evq{{q}}
\def\evr{{r}}
\def\evs{{s}}
\def\evt{{t}}
\def\evu{{u}}
\def\evv{{v}}
\def\evw{{w}}
\def\evx{{x}}
\def\evy{{y}}
\def\evz{{z}}

% Matrix
\def\mA{{\bm{A}}}
\def\mB{{\bm{B}}}
\def\mC{{\bm{C}}}
\def\mD{{\bm{D}}}
\def\mE{{\bm{E}}}
\def\mF{{\bm{F}}}
\def\mG{{\bm{G}}}
\def\mH{{\bm{H}}}
\def\mI{{\bm{I}}}
\def\mJ{{\bm{J}}}
\def\mK{{\bm{K}}}
\def\mL{{\bm{L}}}
\def\mM{{\bm{M}}}
\def\mN{{\bm{N}}}
\def\mO{{\bm{O}}}
\def\mP{{\bm{P}}}
\def\mQ{{\bm{Q}}}
\def\mR{{\bm{R}}}
\def\mS{{\bm{S}}}
\def\mT{{\bm{T}}}
\def\mU{{\bm{U}}}
\def\mV{{\bm{V}}}
\def\mW{{\bm{W}}}
\def\mX{{\bm{X}}}
\def\mY{{\bm{Y}}}
\def\mZ{{\bm{Z}}}
\def\mBeta{{\bm{\beta}}}
\def\mPhi{{\bm{\Phi}}}
\def\mLambda{{\bm{\Lambda}}}
\def\mSigma{{\bm{\Sigma}}}

% Tensor
\newcommand{\tens}[1]{\bm{\mathsfit{#1}}}
\def\tA{{\tens{A}}}
\def\tB{{\tens{B}}}
\def\tC{{\tens{C}}}
\def\tD{{\tens{D}}}
\def\tE{{\tens{E}}}
\def\tF{{\tens{F}}}
\def\tG{{\tens{G}}}
\def\tH{{\tens{H}}}
\def\tI{{\tens{I}}}
\def\tJ{{\tens{J}}}
\def\tK{{\tens{K}}}
\def\tL{{\tens{L}}}
\def\tM{{\tens{M}}}
\def\tN{{\tens{N}}}
\def\tO{{\tens{O}}}
\def\tP{{\tens{P}}}
\def\tQ{{\tens{Q}}}
\def\tR{{\tens{R}}}
\def\tS{{\tens{S}}}
\def\tT{{\tens{T}}}
\def\tU{{\tens{U}}}
\def\tV{{\tens{V}}}
\def\tW{{\tens{W}}}
\def\tX{{\tens{X}}}
\def\tY{{\tens{Y}}}
\def\tZ{{\tens{Z}}}

% Graph
\def\gA{{\mathcal{A}}}
\def\gB{{\mathcal{B}}}
\def\gC{{\mathcal{C}}}
\def\gD{{\mathcal{D}}}
\def\gE{{\mathcal{E}}}
\def\gF{{\mathcal{F}}}
\def\gG{{\mathcal{G}}}
\def\gH{{\mathcal{H}}}
\def\gI{{\mathcal{I}}}
\def\gJ{{\mathcal{J}}}
\def\gK{{\mathcal{K}}}
\def\gL{{\mathcal{L}}}
\def\gM{{\mathcal{M}}}
\def\gN{{\mathcal{N}}}
\def\gO{{\mathcal{O}}}
\def\gP{{\mathcal{P}}}
\def\gQ{{\mathcal{Q}}}
\def\gR{{\mathcal{R}}}
\def\gS{{\mathcal{S}}}
\def\gT{{\mathcal{T}}}
\def\gU{{\mathcal{U}}}
\def\gV{{\mathcal{V}}}
\def\gW{{\mathcal{W}}}
\def\gX{{\mathcal{X}}}
\def\gY{{\mathcal{Y}}}
\def\gZ{{\mathcal{Z}}}

% Sets
\def\sA{{\mathbb{A}}}
\def\sB{{\mathbb{B}}}
\def\sC{{\mathbb{C}}}
\def\sD{{\mathbb{D}}}
% Don't use a set called E, because this would be the same as our symbol
% for expectation.
\def\sF{{\mathbb{F}}}
\def\sG{{\mathbb{G}}}
\def\sH{{\mathbb{H}}}
\def\sI{{\mathbb{I}}}
\def\sJ{{\mathbb{J}}}
\def\sK{{\mathbb{K}}}
\def\sL{{\mathbb{L}}}
\def\sM{{\mathbb{M}}}
\def\sN{{\mathbb{N}}}
\def\sO{{\mathbb{O}}}
\def\sP{{\mathbb{P}}}
\def\sQ{{\mathbb{Q}}}
\def\sR{{\mathbb{R}}}
\def\sS{{\mathbb{S}}}
\def\sT{{\mathbb{T}}}
\def\sU{{\mathbb{U}}}
\def\sV{{\mathbb{V}}}
\def\sW{{\mathbb{W}}}
\def\sX{{\mathbb{X}}}
\def\sY{{\mathbb{Y}}}
\def\sZ{{\mathbb{Z}}}

% Entries of a matrix
\def\emLambda{{\Lambda}}
\def\emA{{A}}
\def\emB{{B}}
\def\emC{{C}}
\def\emD{{D}}
\def\emE{{E}}
\def\emF{{F}}
\def\emG{{G}}
\def\emH{{H}}
\def\emI{{I}}
\def\emJ{{J}}
\def\emK{{K}}
\def\emL{{L}}
\def\emM{{M}}
\def\emN{{N}}
\def\emO{{O}}
\def\emP{{P}}
\def\emQ{{Q}}
\def\emR{{R}}
\def\emS{{S}}
\def\emT{{T}}
\def\emU{{U}}
\def\emV{{V}}
\def\emW{{W}}
\def\emX{{X}}
\def\emY{{Y}}
\def\emZ{{Z}}
\def\emSigma{{\Sigma}}

% entries of a tensor
% Same font as tensor, without \bm wrapper
\newcommand{\etens}[1]{\mathsfit{#1}}
\def\etLambda{{\etens{\Lambda}}}
\def\etA{{\etens{A}}}
\def\etB{{\etens{B}}}
\def\etC{{\etens{C}}}
\def\etD{{\etens{D}}}
\def\etE{{\etens{E}}}
\def\etF{{\etens{F}}}
\def\etG{{\etens{G}}}
\def\etH{{\etens{H}}}
\def\etI{{\etens{I}}}
\def\etJ{{\etens{J}}}
\def\etK{{\etens{K}}}
\def\etL{{\etens{L}}}
\def\etM{{\etens{M}}}
\def\etN{{\etens{N}}}
\def\etO{{\etens{O}}}
\def\etP{{\etens{P}}}
\def\etQ{{\etens{Q}}}
\def\etR{{\etens{R}}}
\def\etS{{\etens{S}}}
\def\etT{{\etens{T}}}
\def\etU{{\etens{U}}}
\def\etV{{\etens{V}}}
\def\etW{{\etens{W}}}
\def\etX{{\etens{X}}}
\def\etY{{\etens{Y}}}
\def\etZ{{\etens{Z}}}

%%%%%%%%%%%%%%%%%%%%%%%%%%%%%%%%%%%%%%%%%%%%%%%%%%%%%%%%%%%%%%%%%%%%%%%%%%%%%%%%
\begin{abstract}

In this paper, we present a scalable deep learning approach to solve opinion dynamics stochastic optimal control problems with mean field term coupling in the dynamics and cost function. Our approach relies on the probabilistic representation of the solution of the Hamilton-Jacobi-Bellman partial differential equation.  Grounded on the nonlinear version of the Feynman-Kac lemma,  the solutions of the  Hamilton-Jacobi-Bellman partial differential equation are linked to the solution of a set of Forward-Backward Stochastic Differential Equations.  These equations can be solved numerically using a novel deep neural network with architecture tailored to the problem in consideration. The resulting algorithm is tested on a polarized opinion consensus experiment. We showcase the scalability and generalizability of our algorithm on a 10k agent experiment. The proposed framework opens up the possibility for future applications on extremely large-scale problems.

\end{abstract}

%%%%%%%%%%%%%%%%%%%%%%%%%%%%%%%%%%%%%%%%%%%%%%%%%%%%%%%%%%%%%%%%%%%%%%%%%%%%%%%%
\section{INTRODUCTION}
With the fast development of social media and their enormous impact on our societies, there is an increasing interest to gain a deeper understanding of the underlying mechanisms of network-enabled opinion dynamics. An appealing feature of opinion dynamics is the interaction and the exchange of opinion which occurs among individual agents or groups of agents. Hence, from a computational perspective the opinion dynamics can be characterized as simulating the evolution of agents' opinion over time under the interaction with their peers. Agent-based model in which the opinion is represented as a real value, such as DeGroot \cite{degroot1974reaching}, Friedkin and Johnsen (FJ) model \cite{friedkin1990social} and their variants, have achieved great empirical success. Interestingly, the phenomena of emulation, herding behavior and polarization occur when one describes the opinion exchange between agents as a graph \cite{wang2021robust}.

In the recent years, Mean field (MF) game and control have become critical tools for analyzing large-population system. In MF game and control, the interaction between indistinguishable agents are negligible though the aggregated influence accumulated among agents is significant. The corresponding MF social opinion optimal control problem is aiming to optimise the sum of agents' cost, which is also known as social cost \cite{wang2021robust} in order to steer the social opinion \cite{zarezade2017steering}, reach consensus \cite{albi2015optimal}, etc. The research was initiated by \cite{huang2012social} in which the Social Certainty Equivalence method was proposed. The work in \cite{wang2017social} adopted stochastic jump diffusion in the opinion dynamics and developed asymptotic team-optimal solution. In \cite{huang2019linear}, the authors consider the scenario where major agents exist and formulated the problem by forward backward stochastic differential equations theory.

From the perspective of stochastic optimal control theory, the MF game and control can be associated with a Hamilton-Jacobi-Bellman (HJB) equation, which is a nonlinear, second-order Partial Differential Equation (PDE).  Different methodologies exist for solving the aforementioned HJB PDE using sampling. In particular, the Path integral control methodology \cite{kappen2005path} leverages Cole-Hopf transformation to obtain a linear PDE representation of the initial nonlinear HJB PDE. The solution of this linear PDE is given  by the linear Feynman-Kac Theorem \cite{karatzas2012brownian}. While the Path Integral control approach can handle nonlinear stochastic control problems, it relies on assumptions between  control authority and the variance and type of the noise in the dynamics. These assumptions restrict the applicability of  Path Integral control to specific classes of problems. 

Stochastic control methods based on the nonlinear Feynman-Kac lemma \cite{pardoux2014stochastic} overcome the limitations of Path Integral control by representing the HJB PDE as a  system of coupled Forward and Backward Stochastic Differential Equations (FBSDEs). These representations hold for general classes of PDEs that arise in stochastic optimal control.  Recent  work \cite{han2018solving} incorporates deep learning into the FBSDEs representation and demonstrates feasibility  and numerical efficiency. The key idea is to represent the FBSDEs system  with a deep neural network architecture, which can be trained with back-propagation through time. The resulting algorithm is known as \emph{Deep FBSDEs}. The work in \cite{pereira2019learning} and \cite{wang2019deep} improve on the Deep-FBSDE approach by incorporating importance sampling within the FBSDEs representation and using a novel and scalable DNN architecture  based on  Long Short-Term Memory neural networks. The proposed algorithm overcomes  limitations of work in \cite{han2018solving} in terms of scalability and training speed. The framework is further extended to the stochastic differential game using fictitious play in \cite{han2020deep} supported by theoretical analysis \cite{han2020convergence}. \cite{chen2021large} improves the scalability of prior work \cite{han2020deep} by leveraging symmetric property in the mean field game and provides theoretical analysis on the important sampling technique in FBSDEs literature.

In this paper, we study the mean field optimal control problem with graphic neighborhood structure under the context of opinion dynamics. Different from \cite{han2020convergence} and \cite{chen2021large} which have access to global information, we consider local graphic information of the state of neighboring agents. In our work, we consider the MF term existing in the cost functional and dynamics which causes the difficulties of control design \cite{huang2012social}. A novel HJB formulation and its corresponding FBSDEs system, named as Deep Graphic FBSDEs (DG-FBSDE), are proposed. In particular, the main contributions of our work is threefold:

\noindent \textbf{1)} We derive a novel HJB PDE and propose associated DG-FBSDEs to handle the opinion dynamics control when MF term is in the cost functional and state dynamics. The underlying DNN backbone architecture relies on Residual Network with time embedding \cite{chen2021large,de2021diffusion}.  Notably, DG-FBSDEs can be generalized to the form in which, the global information is accessible as presented in the prior work \cite{chen2021large}. 

\noindent \textbf{2)}  We demonstrate the scalability of DG-FBSDEs with a 10K agents experiment and the capability of DG-FBSDEs with a polarized opinion problem.

\noindent \textbf{3)}  We illustrate the generalization capability of DG-FBSDEs as a deep learning model. Our model is tested on larger number of agents with the policy trained with few agents. It paves the way for potential large scale applications in the real world.

%%%%%%%%%%%%%%%%%%%%%%%%%%%%%%%%%%%%%%%%%%%%%%%%%%%%%%%%%%%%%%%%%%%%%%%%%%%%%%%%
\section{NOTATIONS AND PROBLEM FORMULATION}

Let $\rmX$  denotes the random variable representing the state of all agents and $\rmX_i$ denote random variable for $i$th agent. The notation	$\rmX_{-i}$ denotes random variable of the state  all agents except $i$th.  Moreover, let $ V(\rmX,t) $ and twice differentiable value function with $ V_{\rvx \rvx}$ corresponding to the gradient and Hessian  w.r.t $\rmX$

\subsection{Problem Formulation}
The opinion of an individual agent in the networked environment can be influenced by neighbors, and through which, its belief is updated. Here we consider the belief update \cite{xu2020paradox} of a network of N agents  in the form of Friedkin and Johnsen (FJ) model \cite{friedkin1990social}:
\begin{equation}\label{eq:opinion-dyn}
   \begin{aligned}
    \rd \rmX_i(t) &= -\alpha(t)\rmX_i(t)\rd t +\alpha(t)f\left(\sA_i\right)\rd t+\Sigma(\rmX_i,t) \rd \rmW_t^{i} \quad \\
    &\quad \quad \quad \quad \quad \text{s.t}\quad \rmX(t_0)=\rvx_{t_0},
   \end{aligned}
 \end{equation}
where $\rmX_i \in \sR^{n_x}$ is the state (opinion) of representative $i$th agent defined on space $\gX$. $\Sigma: \gX \times [t_0, T] \rightarrow \gX \times \sR^{n_w}$ is diffusion term. The state is driven by $n_w$ dimensional Brownian motion which is denoted as $\rmW=(\rmW^{1}_t,..,\rmW^i_t,...,\rmW^N_t)$. $\alpha:\gX \times [t_0, T] \rightarrow  \sR$ and $f:\gX \times [t_0,T] \rightarrow  \gX$ are susceptibility to local influence, and center of bias \cite{xu2020paradox} respectively. $\sA_i$ defines the set of neighbors of $i$th agent:
   \begin{align*}
      \sA_i = \left\{\rmX_j | \norm{\rmX_i-\rmX_j}_2^2\leq \epsilon \right\},
   \end{align*}
 where $\epsilon$ is the radius of neighborhood.
Based on eq.\ref{eq:opinion-dyn}, we consider finite time horizon N-player stochastic differential game with dynamics,
\begin{equation}\label{eq:orig_prob_local}
  \begin{aligned}
   \rd \rmX_i(t) &= -\alpha(t)\rmX_i(t)\rd t +\alpha(t)f\left(\sA_i\right)\rd t\\
   &+ G(\rmX_i,t)\rmU_i(t)\rd t +\Sigma(\rmX_i,t) \rd \rmW_t^i, \ \ \text{s.t}\ \rmX(t_0)=\rvx_{t_0},\\
  \end{aligned}
\end{equation}
where $\rmU_i \in \gU \subseteq \sR^{n_u}$ is control for representative $i$th agent and $G(t): \mathcal{X} \times [t_0, T] \rightarrow \gX \times \gU$ is the actuator dynamics. Interestingly, eq.\ref{eq:orig_prob_local} can be interpreted as the Mckean-Vlasov Stochastic Differential Equation \cite{carmona2015forward} which is broadly investigated in Mean Field Game (MFG) literature \cite{carmona2013control}.

Following  \cite{wang2021robust}, the objective function for each individual agent is formulated as follows: 
\begin{equation}\label{eq:ind-obj}
   \begin{aligned}
      \gJ_i(\rvx, \rvu) = \mathbb{E} \int_{t_0}^{T}\norm{\rmX_i(t) -\Gamma \rmX^{(\sA)}_i(t)}_Q^{2}+ \frac{1}{2}\norm{\rmU_i(t)}_{R}^2 \rd t\\
      +\mathbb{E}\norm{\rmU_i(T)-\Gamma \rmX_i^{(\sA)}(T)}_Q^2,
   \end{aligned}
 \end{equation}
 where $\Gamma$ is a constant coefficient controlling the importance of neighbor's opinion and  $\rmX_i^{(\sA)}(t)$ is the averaging opinion of neighborhood of $i$th agent defined as follows: 
   \begin{align*}
      \rmX_i^{(\sA)}(t)=\frac{1}{|\sA_i|}\sum_{\rmX_j \in \sA_i} \rmX_j(t).
   \end{align*}
 The social cost is defined by the expression: 
 \begin{equation}\label{eq:soc-obj}
   \begin{aligned}
      \gJ_{soc}(\rmX,\rmU)=\sum_{i=1}^{N}\gJ_{i}(\rmX,\rmU),
   \end{aligned}
 \end{equation}
where $N$ is the total number of agents.
The individual agent is aiming to find the minimum social cost given the local neighbors information. Therefore the optimal control is: 
\begin{equation}\label{eq:ind-soc-obj}
   \begin{aligned}
      \rmU^{*}_i=\argmin_{\rmU_i} \gJ_{soc}(\rmX,\rmU).
   \end{aligned}
 \end{equation}

\section{Decoupling the Objective Function}
Noticing that the social objective function (eq.\ref{eq:soc-obj}) is coupled with all the agents, the actual objective for an individual agent to optimize is not clear. Intuitively, the social cost function contains terms that the representative agent $i$ is unable to affect (i.e. agent $i$ cannot influence the states of agents which are not inside the neighborhood through interaction.). In this section, we are going to transform the original optimization problem (eq.\ref{eq:ind-soc-obj}) into a more favorable form, in which the irrelevant terms will be dropped. One can decouple the social objective function into the terms related and independent w.r.t the representative agent $i$ as in Theorem \ref{thm:graph-obj}.
\begin{theorem}\label{thm:graph-obj}
   The social objective function in \eqref{eq:soc-obj} can be decomposed as:
   \begin{equation}
      \begin{aligned}
         \gJ_{i}(\rmX,\rmU)=&\mathbb{E}\bigg[\int_{t_0}^{T}\gL_i(\rmX,\rmU,t)+\gL_{-i}(\rmX,\rmU,t)\rd t\bigg]\\
         &+\mathbb{E}\left[\gC_i(\rmX,T)+\gC_{-i}(\rmX,T)\right],
      \end{aligned}
    \end{equation}
    where:
    \begin{equation}
      \begin{aligned}
         &\gL_{i}(\rmX_i, \rmU_i,t;\rvx_{-i})=\rmX_i^\T S_1 \rmX_i+S_2\rmX_i+\frac{1}{2}\norm{\rmU_i}_{R}^2,\\
         &\gC_i(\rmX,T;\rmX_{-i})=\rmX_i^\T S_1 \rmX_i+S_2\rmX_i,\\
         &S_1=\bar{s}_1\Gamma^\T Q\Gamma\bar{s}_1+Q,\\
         &S_2=-2\Gamma^\T\left(\rmX^{(\sA)}_i(t) ^\T Q+\bar{s}_2 Q+\bar{s}_3Q\Gamma \right),\\
         &\bar{s}_1=\sqrt{\sum_{\rmX_k\in\sA_i(t)}\frac{1}{|\sA_k(t)|^2}},~~ \bar{s}_2=\sum_{\rmX_k \in \sA_i(t)}\frac{1}{|\sA_{k}(t)|}\rmX_k\\
         & \bar{s}_3=\sum_{\rmX_k\in \sA_{i}(t)}\frac{1}{|\sA_k(t)|}\rmX_k^{(\sA/i)}(t),
      \end{aligned}
    \end{equation}
    where $\rmX_k^{(\sA/i)}(t):=\frac{1}{|\sA_k|}\sum_{\rmX_n \in \sA_{k/i}}\rmX_n$ represents the scaled average of the opinions of $k$th agent's neighbors except agent $i$. $\gL_{-i}$ and $\gC_{-i}$ are independent of agent $i$. 
    % {\color{red} Is there a specific form for these terms?}
\end{theorem}
\begin{proof}
   For individual objective function for agent $i$, one can write the running cost functional at time $t$ as follows:
   \begin{equation}
     \begin{aligned}
       \gJ^{run}_i( \rmX, \rmU, t)  &:= \norm{\rmX_i -\Gamma \rmX^{(\sA)}_i(t)}_Q^{2}+ \frac{1}{2}\norm{\rmU_i}_{R}^2\\
                     &= \underbrace{\norm{\rmX_i}_Q^{2}-2\Gamma \rmX^{(\sA)}_i(t) ^\T Q \rmX_i+\frac{1}{2}\norm{\rmU_i}_{R}^2}_{\phi_i}\\
                     &+\underbrace{\rmX^{(\sA)}_i(t) ^\T\Gamma^\T  Q \Gamma\rmX^{(\sA)}_i(t)}_{\phi^{'}_{i}},
     \end{aligned}
   \end{equation}
   where $\phi_i$ represents for the term related to agent $i$ while $\phi_{i}^{'}$ represents for the terms which are irrelevant. Similarly, for $j$th agent, the running cost reads:
   \begin{align*}
       \gJ^{run}_j(\rmX, \rmU, t)
         &= \norm{\rmX_j -\Gamma \rmX^{(\sA)}_j(t)}_Q^2+\frac{1}{2}\norm{\rmU_j}^2_{R}.\\
   \end{align*}
   If $i$th agent does not appear in $j$th agent's neighbourhood, the relevant term $\phi_i$ will vanish:
   \begin{align*}
       \gJ^{run}_j(\rmX, \rmU, t)
         &= \phi_j+\phi^{'}_{j}=0+\phi^{'}_{j}.
   \end{align*}
   On the contrary, If $i$th agent exists in $j$th agent's neighbourhood, the running cost functional can be decomposed as:
   \begin{equation}
     \begin{aligned}
       &\gJ^{run}_j(\rmX, \rmU, t)\\
           &= \norm{\rvx_j -\Gamma \rvx^{(\sA)}_j}_Q^2+\frac{1}{2}\norm{\rmU_j}^2_{R}\\
           &= \norm{\rmX_j}_Q^{2}-2\Gamma \left(\rmX^{(\sA)}_j\right) ^\T Q \rmX_j+\norm{\Gamma\rvx^{(\sA)}_j}_Q^{2}+\frac{1}{2}\norm{\rmU_j}_{R}^2\\
           &\text{Noticing that $\rvx^{(\sA)}_j=\frac{1}{|\sA_{j}(t)|}\left(\rvx_i+\sum_{k \in \sA_j(t)/i}\rvx_k\right)$}\\
           &= \norm{\rmX_j}_Q^{2}-2\Gamma \left(\frac{1}{|\sA_{j}(t)|}\rvx_i+\rmX_{j}^{(\sA/i)}\right) ^\T Q \rvx_j+\frac{1}{2}\norm{\rmU_j}_{R}^2\\
           &+\Gamma^{2}\norm{\frac{1}{|\sA_{j}(t)|}\rmX_i+\rmX_{j}^{(\sA/i)}}_Q^2 \\
           &= \frac{-2\Gamma}{|\sA_{j}(t)|}\rmX_i^\T Q\rvx_j+\Gamma^2\norm{\frac{1}{|\sA_j(t)|}\rmX_i }_Q^2\\
           &+\frac{2\Gamma^2}{|\sA_{j}(t)|}\rmX_j^{(\sA/i)^\T} Q\rmX_i+\phi^{'}_{j}.\\ 
     \end{aligned}
   \end{equation}
   The cost function $ \gJ^{run}_j(\rvx, \rvu, t)$ can be further written as:
   \begin{equation}
     \begin{aligned}
       &\gJ^{run}_j(\rmX, \rmU, t)\\
         &= \sI\left[\rmX_i \in \sA_j(t)\right](\phi_j)+\phi^{'}_{j}\\
         &=\sI\left[\rmX_i \in \sA_j(t)\right]\left[ \frac{-2\Gamma}{|\sA_{j}(t)|}\rmX_i^\T Q\rmX_j+\Gamma^2\norm{\frac{1}{|\sA_j(t)|}\rmX_i }_Q^2\right..\\
         &+\left. \frac{2\Gamma^2}{|\sA_{j}(t)|}\rmX_j^{(\sA/i)^\T} Q\rmX_i \right]+\phi^{'}_{j},\\
         % &=\sI\left(i \in \sA_j(t)\right)\left[ \frac{-2\Gamma}{|\sA_{j}(t)|}\rvx_i^\T Q\rvx_j+\frac{1}{|\sA_j(t)|^2}\left(\rvx_i^\T \Gamma^\T Q \Gamma \rvx_i\right)+\frac{2}{|\sA_{j}(t)|}\left(\rvx_j^{(A/i)}(t)^\T\Gamma^\T Q \Gamma \rvx_i \right)\right]+\phi^{'}_{j}
     \end{aligned}
   \end{equation}
   where $\sI$ denotes the identification function.
   By rearranging the terms, we can distill the terms which are related to $i$th agents:
   \begin{equation}
     \begin{aligned}
       &\gL_i(\rmX,\rmU,t)=\sum_{k=1}^{N} \phi_k\\
                       &=\norm{\rmX_i}_Q^{2}-2\Gamma \rmX^{(\sA)}_i(t) ^\T Q \rmX_i+\frac{1}{2}\norm{\rmU_i}_{R}^2\\
                       &+\sum_{k \in \sN/i} \sI\left[i \in \sA_k(t)\right]\left[ \frac{-2\Gamma}{|\sA_{k}(t)|}\rvx_i^\T Q\rmX_k\right.\\
                       &\left.+\frac{\Gamma^2}{|\sA_k(t)|^2}\rvx_i^\T Q \rmX_i+\frac{2\Gamma^2}{|\sA_{j}(t)|}\rmX_j^{(\sA/i)^\T} Q\rmX_i\right]\\
                       &=\norm{\rmX_i}_Q^{2}-2\Gamma \rmX^{(\sA)}_i(t) ^\T Q \rmX_i+\frac{1}{2}\norm{\rmU_i}_{R}^2\\
                       &+\sum_{k \in \sA_i(t)}\left[ \frac{-2\Gamma}{|\sA_{k}(t)|}\rvx_i^\T Q\rmX_k\right.\\
                       &\left.+\frac{\Gamma^2}{|\sA_k(t)|^2}\rmX_i^\T Q \rmX_i+\frac{2\Gamma^2}{|\sA_{j}(t)|}\rmX_j^{(\sA/i)^\T} Q\rmX_i\right]\\
                       &\stackrel{(*)}=\norm{\rmX_i}_Q^{2}-2\Gamma \rmX^{(\sA)}_i(t) ^\T Q \rmX_i+\frac{1}{2}\norm{\rmU_i}_{R}^2\\
                       &-2\Gamma\left(\sum_{k \in \sA_i(t)}\frac{1}{|\sA_{k}(t)|}\rmX_k\right)^\T Q \rmX_i\\
                       &+\Gamma^2\left(\sum_{k\in\sA_i)(t)}\frac{1}{|\sA_k(t)|^2}\right)\rmX_i^\T Q  \rmX_i\\
                       &+2\Gamma^2\left(\sum_{k\in \sA_{i}(t)}\frac{1}{|\sA_k(t)|}\rmX_k^{(\sA/i)}(t)\right)^\T Q \rmX_i\\
                       &=\rmX_i^\T \left(\bar{s_1}\Gamma^\T Q\Gamma\bar{s_1}+Q\right)\rmX_i\\
                       &-2\Gamma^\T\left(\rmX^{(\sA)}_i(t) ^\T Q+\bar{s}_2 Q-\bar{s}_3Q\Gamma \right)\rmX_i+\frac{1}{2}\norm{\rmU_i}_{R}^2\\
                       &=\rmX_i^\T S_1 \rmX_i+S_2^\T\rmX_i+\frac{1}{2}\norm{\rmU_i}_{R}^2,\\
     \end{aligned}
   \end{equation}
   where
   \begin{equation}
     \begin{aligned}
       S_1&=\bar{s_1}\Gamma^\T Q\Gamma\bar{s_1}+Q\\
       S_2&=-2\Gamma^\T\left(\rmX^{(\sA)}_i(t) ^\T Q+\bar{s}_2 Q-\bar{s}_3Q\Gamma \right)^\T\\
       \bar{s}_1&=\sqrt{\sum_{k\in\sA_i)(t)}\frac{1}{|\sA_k(t)|^2}}\\
       \bar{s}_2&=\sum_{k \in \sA_i(t)}\frac{1}{|\sA_{k}(t)|}\rmX_k\\
       \bar{s}_3&=\sum_{k\in \sA_{i}(t)}\frac{1}{|\sA_k(t)|}\rmX_k^{\sA/i}(t).
     \end{aligned}
   \end{equation}
  Equality $(*)$ holds because all agents share identical neighborhood radius $\epsilon$. By repeating the algebra above, The terminal state cost functional related to agent $i$ can be easily obtained as:
  % &\gC_i(\rvx,T;\rvx_{-i}(t))=\rvx_i^\T S_1 \rvx_i+S_2\rvx_i,\\
  \begin{equation}
    \begin{aligned}
      \gC_i(\rmX,T;\rmX_{-i})=\rmX_i^\T S_1 \rmX_i+S_2\rmX_i.
    \end{aligned}
  \end{equation}
 \end{proof}
 \begin{figure*}%
  \centering
  \subfloat{{\includegraphics[width=0.47\linewidth]{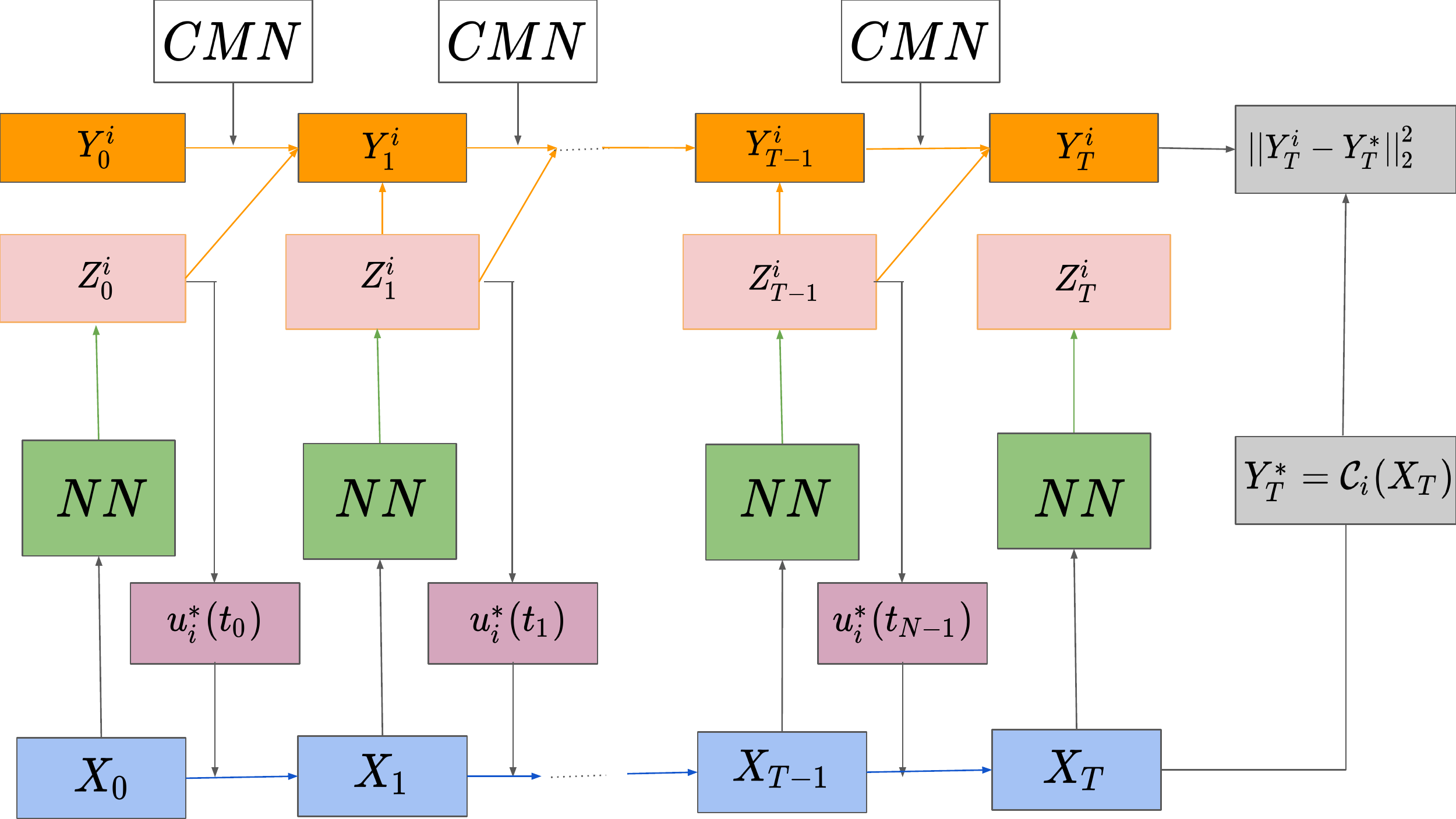} }}%
  \qquad
  \subfloat{{\includegraphics[width=0.47\linewidth]{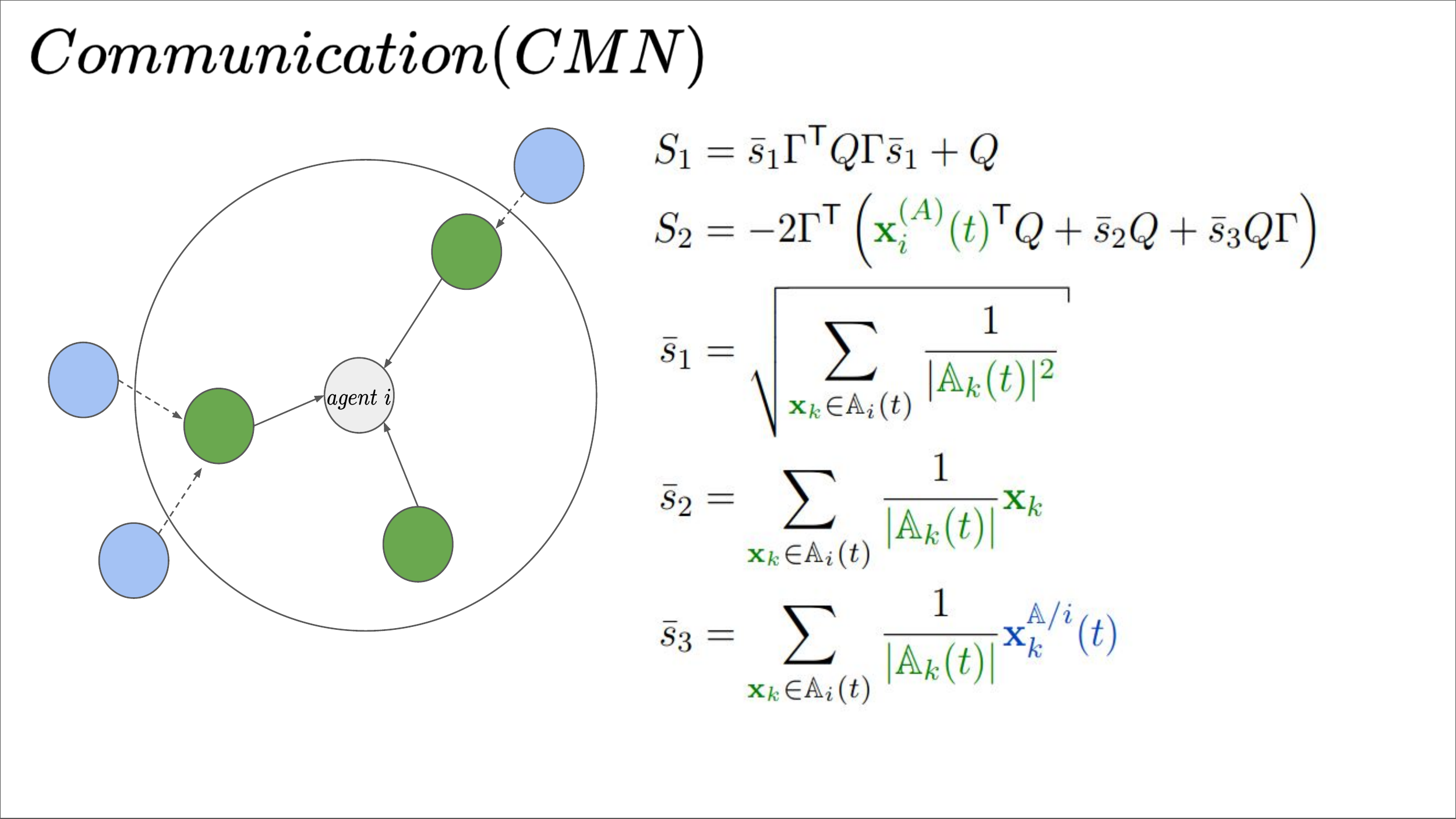} }}%
  \caption{\textbf{Left}: DG-FBSDE network structure. NN represents for neural network backbone. Orange and blue path represent for the propagation of FSDE and BSDE in eq.\ref{eq:IS-FBSDE}. \textbf{Right}: Communication module (CMN) shows information transportation among neighbour agents. Green circles are neighbors of the gray agent, and blue circles are neighbors of green agents which are neighbor's neighbors of agent $i$. The corresponding information in the equations are in the same color.}\label{fig:NN-arch}
\end{figure*}
\section{HJB and FBSDE formulation}
Thanks to Theorem \ref{thm:graph-obj}, the optimal control for each agent can be obtained by solving the simplified optimization problem as follows:
\begin{equation}
   \begin{aligned}
     \rvu^{*}_{i} &=\argmin_{\rmU_i} \gJ_{soc}(\rmX,\rmU)\\
                  &=\argmin_{\rmU_i} \mathbb{E}\int_{t_0}^{T}\gL_i(\rmX,\rmU,t)+\gL_{-i}(\rmX,\rmU,t)\rd t\\
                  &\quad \quad \quad \quad \quad +\mathbb{E}\left[\gC_i(\rmX,T)+\gC_{-i}(\rmX,T)\right]\\
                  &=\argmin_{\rmU_i} \mathbb{E}\int_{t_0}^{T}\gL_i(\rmX,\rmU,t)\rd t+\mathbb{E}\left[\gC_i(\rmX,T)\right]\\                  
                  &\neq \argmin_{\rmU_i} \gJ_i(\rmX,\rmU).
   \end{aligned}
 \end{equation}
\begin{remark}
   The new formulation is optimizing the social cost. On the contrary, the last equation (eq.\ref{eq:ind-obj}) only considers individual interest.
\end{remark}
Hence, the value function for individual agent can be naturally defined as,
\begin{equation}
   \begin{aligned}
     V_t^i(\rmX,t)  &=\inf_{\rmU_i}\mathbb{E}\int_{t_0}^{T}\gL_i(\rmX,\rmU,t)\rd t+\mathbb{E}\left[\gC_i(\rmX,T)\right].\\
   \end{aligned}
 \end{equation}
By leveraging the  stochastic optimal control theory, $V_i$  satisfies  the HJB Equation specified as follows:
\begin{equation}\label{eq:HJB}
  \begin{aligned}
    V^i_t+H(x, V_{\rvx}^i, V_{\rvx\rvx}^{i},t)=0, \quad V(\rvx,T)=\gC_i(\rvx,T),
  \end{aligned}
\end{equation}
where $H$ is known as Hamiltonian in literature. Recalling realized opinion dynamics with controls (eq.\ref{eq:orig_prob_local}), here we denote $F(\rmX_i):=-\alpha(t)\rmX_i(t)\rd t +\alpha(t)f\left(\sA_i\right)$ for simplicity. Knowing the dynamics and objective, the Hamiltonian reads:
\begin{equation}
  \begin{aligned}
    &H(\rvx,V_{\rvx}^i, V_{\rvx\rvx}^{i},t)\\
    &=  \min_{\rvu_i}\left[\gL_i+V_{\rvx}^{i}(F+G \rmU_i)+\frac{1}{2}\Tr(V_{\rvx\rvx}^i\Sigma\Sigma^\T) \right]\\
    &=  \min_{\rvu_i}\left[\left( \rmX_i^\T S_1 \rmX_i+S_2\rmX_i+\frac{1}{2}\norm{\rmU_i}_{R}^2 \right)\right.\\
    &\quad \quad \quad \quad +\left. V_{\rvx}^{i}(F+G\rmU_i)+\frac{1}{2}\Tr(V_{\rvx\rvx}^i\Sigma\Sigma^\T) \right].\\
  \end{aligned}
\end{equation}
By forcing the derivative of the term in the bracket to be zero, one can find the optimal control as:
  \begin{align*}
    \rmU^{*}_i=R_{1}^{-1}G^\T V_{\rvx}^{i}.
  \end{align*}
Plugging the Hamiltonian with optimal control $\rvu^{*}_i$ back to (eq.\ref{eq:HJB}), we can derive the HJB equation for the $i$th representative agent:
\begin{equation}\label{eq:HJB-with-OC}
  \begin{aligned}
    &V_t^i+\frac{1}{2}\Tr(V_{\rvx\rvx}^i\Sigma\Sigma^\T)+V_{\rvx}^{i^\T} F-\frac{1}{2}||\rmU_i^{*}||_2^2\\
    &+\rmX_i^\T S_1 \rmX_i+S_2\rmX_i=0 \\
    &\text{s.t} \ \  V^{i}(\rmX,T)=\gC_i(\rmX,T).\\
  \end{aligned}
\end{equation}
\subsection{FBSDE theory and Important sampling (IS)}
The HJB PDE (eq.\ref{eq:HJB-with-OC}) can be related to a set of FBSDEs by applying the nonlinear Feynman-Kac Lemma \cite{karatzas2012brownian}:
\begin{equation}
   \begin{aligned}
     &\rd \rmX_i=F\rd t+\Sigma \rd \rmW_t^i, \quad \quad \quad \quad \quad \quad \quad \quad \quad  \ \ \quad \quad  \quad \text{(FSDE)},\\
     &\rd \rmY_i=-(\rmX_i^\T S_1 \rmX_i+S_2^\T\rmX_i-\frac{1}{2}||\rmU_i^{*}||_2^2)\rd t+\rmZ_t^{i} \rd \rmW_t   \text{(BSDE)},\\
     &\text{s.t} \quad \rmX_{t_0}=\rvx_{t_0}, \rmY_{T}=\gC_i(\rmX_T).
   \end{aligned}
 \end{equation}
 Here we follow classical notations appeared in FBSDE theory literature: $\rvY_i \equiv V_i$ is value function and  $\rmZ_t^{i}\equiv \Sigma^\T V_{\rvx}^{i}$ is known as the adjoint state. In the FBSDE formulation, the solution of backward process (BSDE) corresponds to the solution of HJB equation. Inspired by \cite{exarchos2018stochastic}, the FBSDE with important sampling (IS) can be further established as:
 \begin{equation}\label{eq:IS-FBSDE}
  \begin{aligned}
    &\rd \rmX_i=(F+G\bar{\rmU}_i)\rd t+\Sigma \rd \rmW_t, \quad \quad \quad \quad \quad \quad \quad  \text{(FSDE)},\\
    &\rd \rmY_i=\left[-(\rmX_i^\T S_1 \rmX_i+S_2^\T\rmX_i-\frac{1}{2}||\rmU_i^{*}||_2^2)+\rmZ_t^{i^\T} \Phi \bar{\rmU}\right]\rd t\\
    &\quad \quad \quad \quad \quad \quad \quad \quad \quad \quad \quad \quad +\rmZ_t^{i^\T} \rd \rmW_t  \quad \quad \quad \ \   \text{(BSDE)},\\
    &\text{s.t} \quad \rmX_{t_0}=\rvx_{t_0}, \rmY_{T}=\gC_i(\rmX_T),
  \end{aligned}
\end{equation}
where we assume $G(\rmX,t)$ can be decompose as $G(\rmX,t)=\Sigma(\rmX,t)\Phi(\rmX,t)$ \cite{pereira2019learning}. The IS formulation allows us to modify the FSDE, while the BSDE still solves the original HJB PDE (eq.\ref{eq:HJB-with-OC}) almost surely. Due to the mathematical flexibility, nominal control policy $\bar{\rmU_i}$ can be any arbitrary control. Here, we adopt the control policy estimated in the previous run/iteration of the algorithm. Theoretically supported by \cite{chen2021large}, IS is able to increase exploration which is critical for value function estimation on non-equilibrium states.
\begin{algorithm}[h]
  \caption{Deep Graphic FBSDE}
  \begin{algorithmic}[1]\label{algo}
  \STATE \textbf{Hyper-parameters}: $N$: Number of players; $T$: Number of timesteps; $N_{gd}$: Number of gradient descent steps; $B$: Batch size; 
  \STATE \textbf{Parameters}: $\phi$: Network weights for Initial Value (IV) prediction $\rmY^i_{0}=f_{\rmY}(\cdot;\phi)$; $\theta$: Network weights for prediction $\rmZ^i_t=f_{\rmZ}(\cdot;\theta)$.
  %   \STATE Function: OptimalControl($V_x$,$\mX$):
  %   \STATE \ \ \ $u^{\star}=\underset{u}{\mathrm{argmin}}\left[C^i(t,\mX(t),\vu(s,\mX(t)))+V_x^{i\T}b\right]$
  %   \STATE \ \ \ return $u^{\star}$
  \STATE Initialize trainable papermeters:$\theta^{0}$, $\phi^{0}$
  \STATE Generate $B$ sample $\rvx_0$ and $B\times T$ Noise.
  \STATE \textbf{repeat}
  \STATE Running parallel for (sub)superscript $i$.
  \FOR{$l \leftarrow 0$ to $N_{gd}-1$}
  \FOR{$t \leftarrow 0$ to $T-1$}
  \IF{$t=0$}
  \STATE Sampling $B$ number of initial states $\rvx_0$.
  \STATE Predict value function for $i$th player: $\rmY^i_0=f_{\rmY}(\rmX_0;\phi)$
  \ELSE
  \STATE Compute network prediction: $\rmZ^i_t=f_{\rmZ}(\rmX_t;\theta)$
  \ENDIF
  \STATE Compute $i$th optimal control:$\rmU^{*}_i=R^{-1}\Phi^\T \rmZ^i_t$
  \STATE Propagate FSDE in (eq.\ref{eq:IS-FBSDE}):
  \STATE $\rd \rmX_i=(F(\rmX_t,t)+G(\rmX_t,t)\rmU^{*}_i)\rd t+\Sigma \rd \rmW_t^i$ 
  \STATE Propagate BSDE in (eq.\ref{eq:IS-FBSDE}):
  \STATE $\rd \rmY_i=[-(\rmX_i^\T S_1 \rvx_i+S_2^\T\rmX_i-\frac{1}{2}||\rmU_i^{*}||_2^2) $ \\
   ~~~~~ $+ \rmZ_t^{i^\T} \Phi \rmU ] \rd t+\rmZ_t^{i^\T} \rd \rmW_t$ 
  \ENDFOR
  \STATE Compute True terminal value $\rmY^{*}_i(\rmX,T)=\gC_i(\rmX_T,T)$
  \STATE Compute loss $\norm{\rmY^{*}_i(\rmX,T)-\rmY_{i}(\rmX,T)}$
  \STATE Gradient Update: $\theta,\phi$
  % \FOR{all $i \in \sI$ \textit{in parallel}}
  % \STATE Collect opponent agent's policy which is same as $i$th policy: $f_{LSTM_{i}}^{m-1}(\cdot), f_{FC_{i}}^{m-1}(\cdot)$
  % \FOR{$l \leftarrow 1$ to $N_{gd}$}
  \ENDFOR
  \STATE \textbf{until} convergence
\end{algorithmic}
\end{algorithm}
\begin{figure}%
  \centering
  \includegraphics[width=0.8\linewidth]{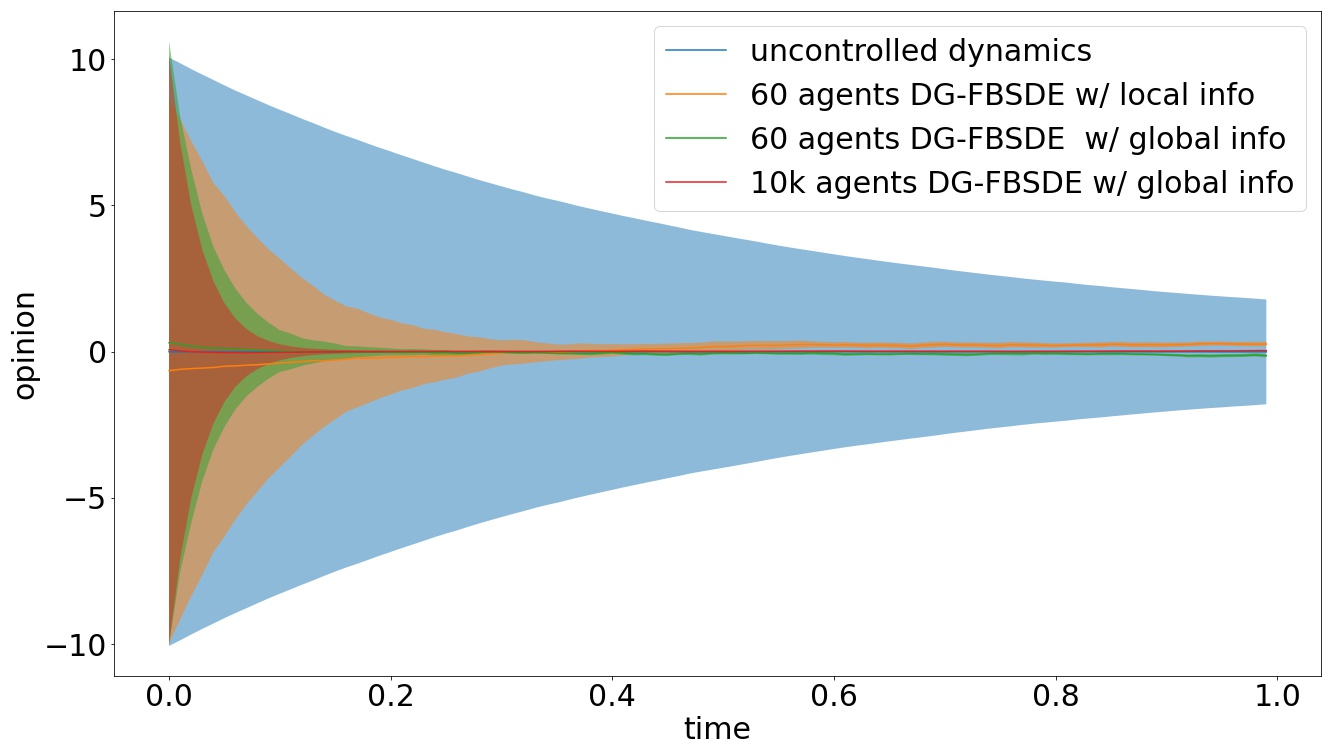} 
  \caption{Validation of our DG-FBSDE on different levels of information. The solid line and shaded color are mean and variance of agents' states. The legend are ordered as the increasing information level in which 10K agents with global information has richest information of the system. One can notice the converge rate of consensus has positive correlation with the information level.}\label{fig:MF_compare}
\end{figure}
\section{Deep Graphic FBSDE (DG-FBSDE)}
\begin{figure*}[h]%
  \centering
  \hspace*{-3.7cm} 
  \includegraphics[width=2.8\columnwidth]{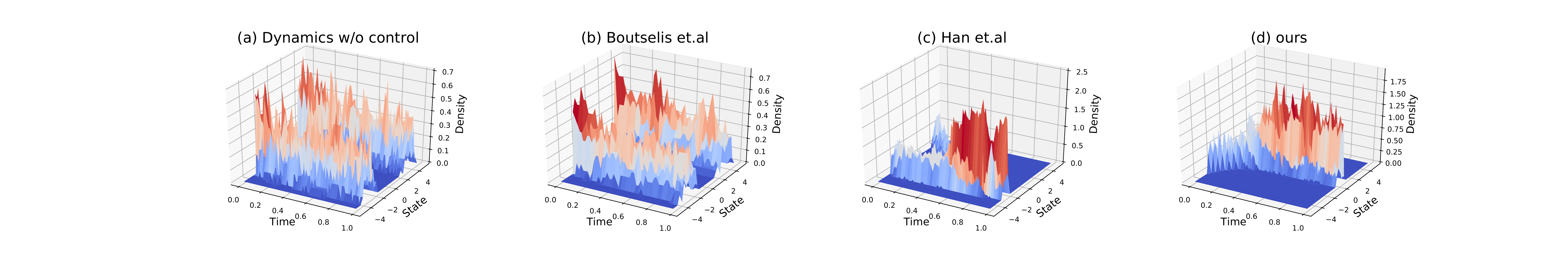}
  \caption{Opinion trajectories with different policies. \textbf{(a)}: The agents fail to reach consensus without controls. \textbf{(b)}:\cite{boutselis2020constrained} cannot drive agents to reach agreement because of the information limitation. \textbf{(c)}:\cite{han2020convergence} can drive agent to the consensus, but it is heavily biased. \textbf{(d)}: When the dynamics equipped with DG-FBSDE, the agents are able to reach neutral consensus.}\label{fig:bimodal}
\end{figure*}
Utilizing aforementioned results, we propose a novel Deep Graphic FBSDE (DG-FBSDE) framework to solve stochastic optimal control problems with local graphic neighborhood information. Notably, the framework can be extended to the system with global information smoothly.

\textbf{Algorithm:} The algorithm is designed in the continuous time horizon $0\leq t \leq T$ during which FSDE and BSDE are propagated via Euler integration scheme (see fig.\ref{fig:NN-arch}: Orange and blue path are FSDE and BSDE respectively). In Algorithm.\ref{algo}, the neural network is aiming to approximate initial value function $\rmY_0^{i}$ and $\rmZ_t^{i}$ component at each timesteps. The loss function for the neural network is defined as the $\gL_2$ norm of the difference between predicted value function $\rmY_T^{i}$ induced by BSDE, and the true terminal value function $\rmY_T^{i*}$ determined by terminal cost function $\gC_i$. The neural network can be trained by standard gradient based optimizers such as Adam \cite{kingma2014adam}.

\textbf{Network Architecture:}
The network architecture is illustrated in fig.\ref{fig:NN-arch}. We use Residual Neural Network with time embedding \cite{chen2021likelihood} as the  neural network backbone. The neural network is sharing the same parameters over timesteps. In this work we consider  agents that are symmetric. This  means that they have identical objective functions (eq.\ref{eq:soc-obj}) and dynamics (eq.\ref{eq:orig_prob_local}). We assign all the agents to the same neural network architecture and then follow the process of   Centralized Training and Decentralized Executing (CTDE) scheme that is used in the Reinforcement Learning literature. By eliminating potentially redundant parameters, we only need to maintain one neural network in fig.\ref{fig:NN-arch} for all agents which paves the way for generalizing trained policy to larger number of agents without additional training.
\section{SIMULATION RESULTS}
\begin{figure}%
  \centering
  \includegraphics[width=\linewidth]{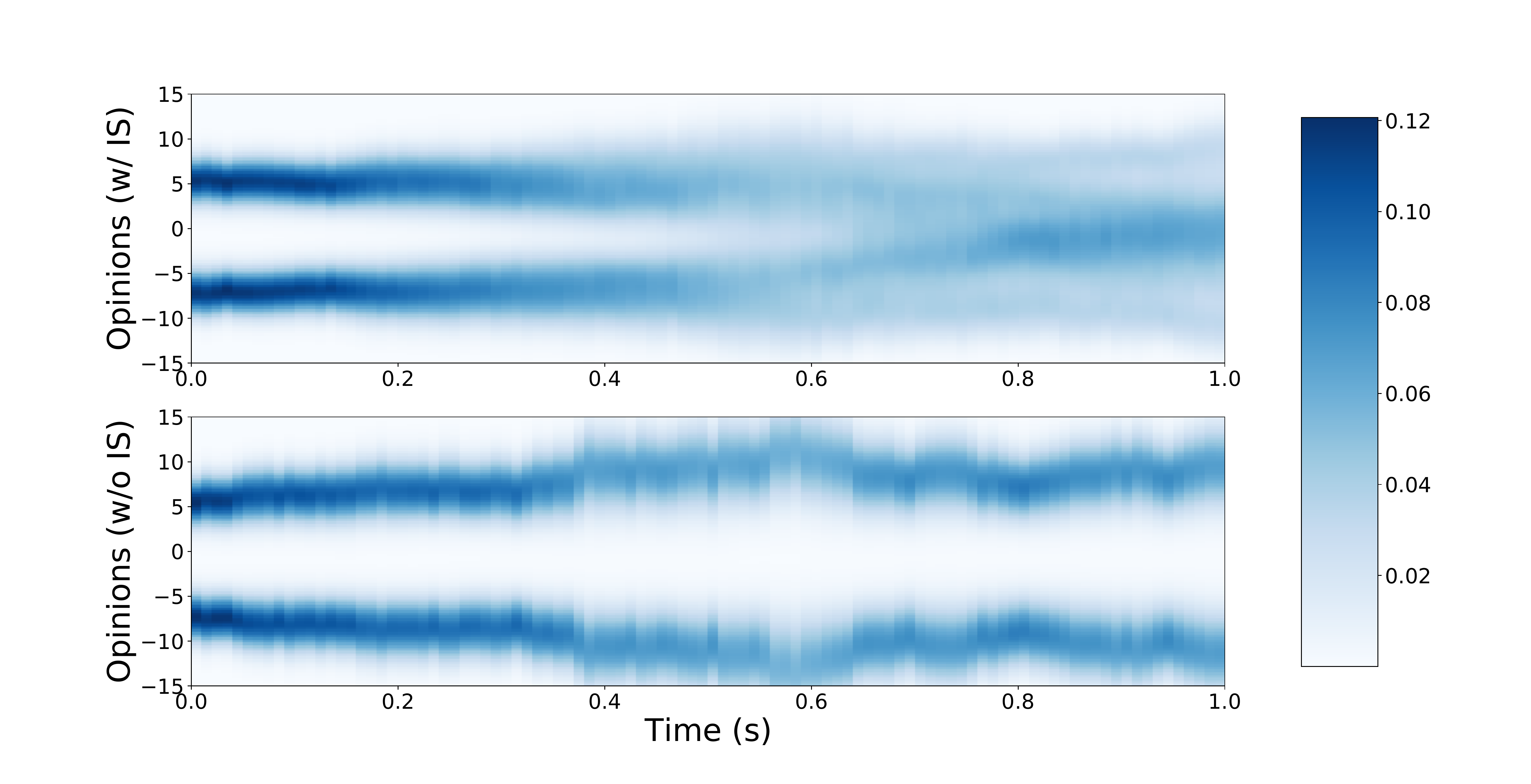} 
  \caption{The difference of sampled opinion region between opinion dynamics with an without important sampling. \cite{han2020convergence} fails to sample neutral regions. By contrast, our model with important sampling is able to cover broader state space which potentially lead to favorable neutral consensus policy.}\label{fig:IS}
\end{figure}
\begin{figure}[h]%
  \centering
  \includegraphics[width=\linewidth]{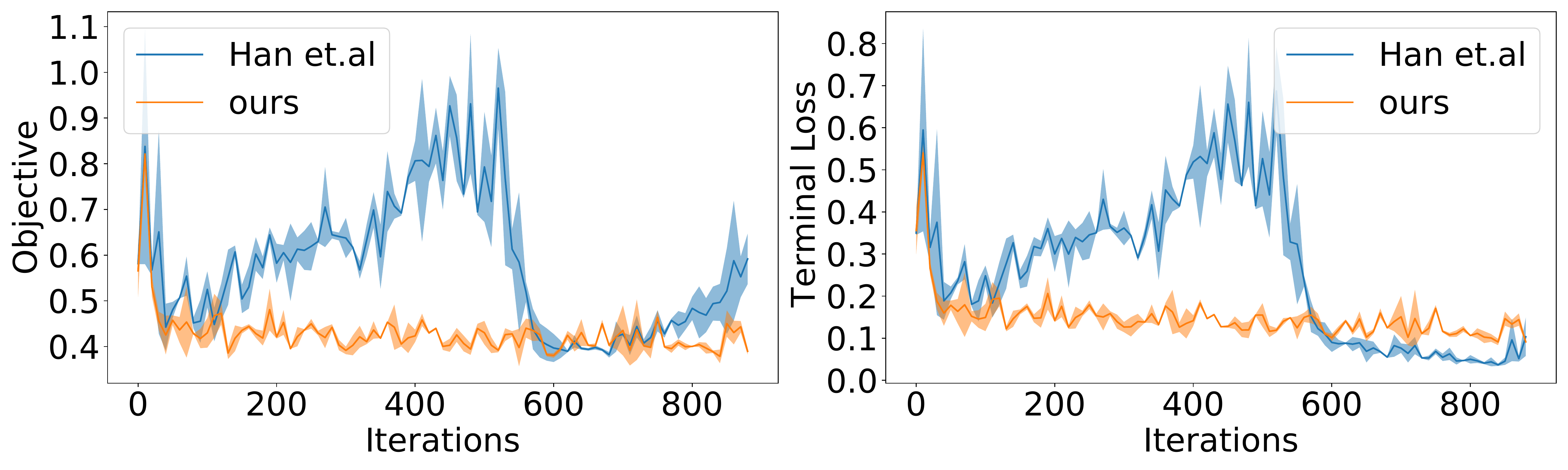}
  \caption{Comparison of baseline \cite{han2020convergence} and our algorithm w.r.t objective (eq.\ref{eq:soc-obj}) and terminal loss. Our model shows more stable training and is able to find a neutral consensus even though both algorithm have close minima.}\label{fig:bimodal_loss}
\end{figure}
\textbf{Setup:} We tested our algorithm on three different opinion scenarios to illustrate that agents will reach consensus uniformly under control induced by DG-FBSDE. For the unimodal Gaussian initialization, we empirically show the applicability of our algorithm for different scopes of information, and surprisingly, DG-FBSDE can be scaled up to 10k agents with global information. For the polarized opinion case, in which opinion states are initialized as bi-modal distribution, the system will not reach consensus without controls. We demonstrate that our algorithm is able to obtain a more favorable \emph{neutral} consensus policy compared  to the work \cite{han2018solving} which uses PDEs solver leveraging the power of deep learning  as well as the work in \cite{boutselis2020constrained}  that is  based on stochastic optimization methods applied to free-energy types of cost functions. Lastly, we show the superior generalization ability of our algorithm by deploying trained policy to larger number of agents \emph{without} further fine-tuning under four-modal opinion initialization, and demonstrate how the controlled agents influence the dynamics of uncontrolled agents. The resulting plots are averaging over 3 repeated independent runs. Solid line and shadow region are mean and standard deviation. The hyperparameters are $\alpha=1$, $\epsilon=2$, $\Sigma=I_N$ unless otherwise noted. For unimodal case, the initial states are sampled from uniform distribution with range of $[-10,10]$. For bimodal case, the initial states are sampled uniformly in the range of $[-8,-4]$ and $[4,8]$ in order to guarantee the agent from the other modal distribution will not appear in the neighborhood. For four modal case, the initial states are sampled uniformly from the range of $[-4,-3]$, $[-1,0]$, $[1.5,2.5]$, and $[4,5]$. The center of bias term is defined as $f(\sA_i)=1/|\sA_i|\sum_{x_{k} \in \sA_{i}} x_k$.
\subsection{Unimodal Consensus}
We first consider the unimodal Gaussian initialization for agents' opinions, and the dynamics of opinions is following eq.\ref{eq:orig_prob_local}. The evolution of opinions without control (eq.\ref{eq:opinion-dyn}) is shown in fig.\ref{fig:MF_compare} in bluer. One can notice that opinions reach consensus gradually due to the interaction component $f(\sA_i)$ in the dynamics. The opinion arrives agreement rapidly by applying DG-FBSDE controls with local information. When the DG-FBSDE has the access to the global information, such as other agents' opinion, the consensus rate will be increased accordingly as being shown in fig.\ref{fig:MF_compare}. Our algorithm can scale up to 10K agents while maintaining superior performance due to the data-driven deep neural network model.

\subsection{Consensus in the phenomenon of Polarization}
In the real world, the phenomenon of polarization among agents are attracting increasing attention \cite{gaitonde2021polarization}. In this section, we consider a polarized opinion scenario in which, the neighborhood radius $\epsilon$ is set to $2$ such that agents are unable to reach agreement without controls. In fig.\ref{fig:bimodal}(a), the opinions of agents are not able to reach consensus since the influence of opinion from the other modal distribution will never present in the $\epsilon$-neighborhood. As a result, \cite{boutselis2020constrained} cannot find a policy which drives agents to the agreement shown in fig.\ref{fig:bimodal}.(b). Prior work \cite{han2018solving} can find the consensus policy which is shown in fig.\ref{fig:bimodal}.(c), but the terminal consensus state is largely biased to one of the opinion. Despite of the difficulties of the problem formulation, DG-FBSDE still manages to induce a neutral consensus policy which is demonstrated in fig.\ref{fig:bimodal}.(d). fig.\ref{fig:bimodal_loss} illustrates that, even though both \cite{han2018solving} and our model is able to minimize the objective function, our model can find a more favorable neutral policy. We suspect that the performance difference is introduced by the important sampling of FBSDEs. In fig.\ref{fig:IS}, one can notice that the dynamics without important sampling is unable to explore the neutral consensus region (in the middle). On the contrary, dynamics with IS has better exploration which is including the desirable consensus region.
\begin{figure}[h]%
  \centering
  \includegraphics[width=\linewidth]{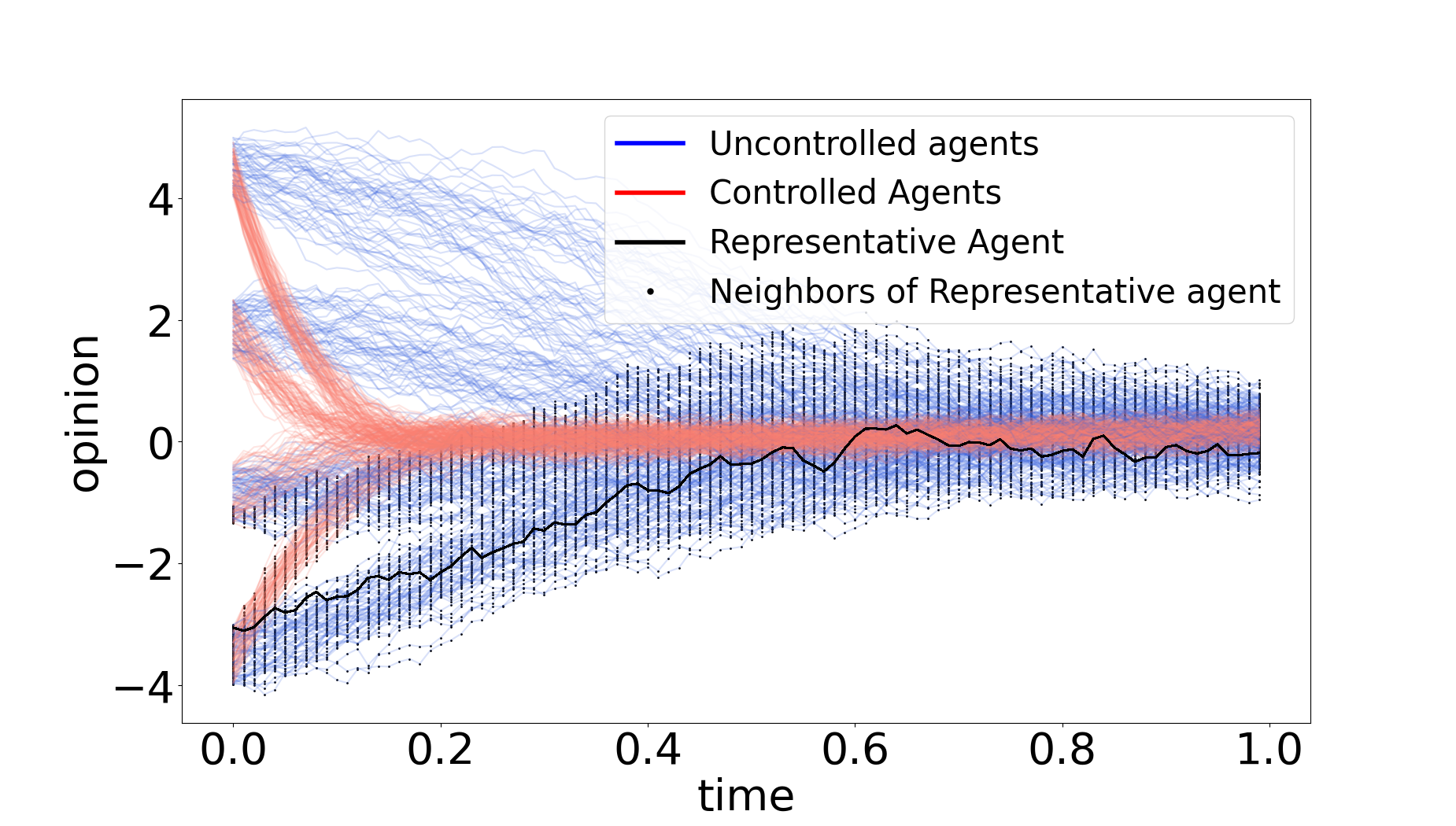}
  \caption{State trajectories demonstrating generalization capability of our model. 250 controlled agents are using DG-FBSDE policy which is trained with 60 agents. The system is able to reach the consensus even half (250) of agents are uncontrolled.}\label{fig:four_mode}
\end{figure}
\subsection{Generalization of DG-FBSDE}
Neural networks are well-known for their generalization ability. In this section, we show that, the policy is applicable to larger number of agents even though it is trained by only small number agents. We use 60 agents in the training phase to obtain the DG-FBSDE policy. Motivated by the scenario where some of agents are not directly controllable in the real world, We simulate the total of 500 agents but half of them are uncontrollable. The another half of agents are equipped with DG-FBSDE policy trained on a 60 agents example. In fig.\ref{fig:four_mode}, one can find that the controlled agents can reach the consensus as expected and the uncontrolled dynamics is merging to the consensus opinion by the influence of neighbors. This experiment verifies the generalization of our model and paves the way for the future in extremely large scale applications. In practice, We can instead to train on smaller number of agents and drop the expensive training when the number of agents is huge.

\section{Conclusion}
In this work, we first derive the HJB PDE when MF term with local information appears  in the state cost functional and dynamics in consideration. We propose the novel DG-FBSDEs algorithm to solve such HJB PDE and verify the performance of our algorithm in different experiment setups. Our algorithm achieves  favorable performance when the phenomena of polarization occurs compared with \cite{han2020convergence} and \cite{boutselis2020constrained}. We showcase that the proposed algorithm generalizes to the case when the global information is available for which cases the consensus rate is accelerated. Lastly, as a deep learning framework, our model is able to generalize to larger number of agents without further fine-tuning. These last property creates opportunities  for  future practical applications of the proposed algorithm to extremely large scale opinion dynamics models.
\addtolength{\textheight}{-3cm}   % This command serves to balance the column lengths

%%%%%%%%%%%%%%%%%%%%%%%%%%%%%%%%%%%%%%%%%%%%%%%%%%%%%%%%%%%%%%%%%%%%%%%%%%%%%%%%
\section*{ACKNOWLEDGMENTS}

This work is supported by the DoD Basic Research Office Award HQ00342110002

%%%%%%%%%%%%%%%%%%%%%%%%%%%%%%%%%%%%%%%%%%%%%%%%%%%%%%%%%%%%%%%%%%%%%%%%%%%%%%%%
\bibliographystyle{unsrt}
\bibliography{reference.bib}
% The authors gratefully acknowledge the contribution of National Research Organization and reviewers' comments.

% %%%%%%%%%%%%%%%%%%%%%%%%%%%%%%%%%%%%%%%%%%%%%%%%%%%%%%%%%%%%%%%%%%%%%%%%%%%%%%%%

% References are important to the reader; therefore, each citation must be complete and correct. If at all possible, references should be commonly available publications.

% \begin{thebibliography}{99}

% \bibitem{c1}
% J.G.F. Francis, The QR Transformation I, {\it Comput. J.}, vol. 4, 1961, pp 265-271.

% \bibitem{c2}
% H. Kwakernaak and R. Sivan, {\it Modern Signals and Systems}, Prentice Hall, Englewood Cliffs, NJ; 1991.

% \bibitem{c3}
% D. Boley and R. Maier, "A Parallel QR Algorithm for the Non-Symmetric Eigenvalue Algorithm", {\it in Third SIAM Conference on Applied Linear Algebra}, Madison, WI, 1988, pp. A20.

% \end{thebibliography}

\end{document}